\documentclass[10pt, conference]{IEEEtran}

\usepackage{amsmath,amsthm,epsfig,amssymb,verbatim,amsopn,cite,subfigure,multirow}
\usepackage{balance,nopageno}
\usepackage[usenames,dvipsnames]{color}
\usepackage[all]{xy}  
\usepackage{url}
\usepackage{amsfonts}
\usepackage{amssymb}
\usepackage{epsfig}
\usepackage{epstopdf}
\usepackage{slashbox}
\usepackage{bm}
\usepackage{balance}
\usepackage[margin=18.0mm,top=18mm,bottom=19mm]{geometry}

\newtheorem{Lemma}{Lemma}

  {\proof}{\proofend}
\newtheorem{proposition}{Proposition}

\DeclareMathOperator*{\argmax}{arg\,max}

\newcommand{\ve}[1]{\boldsymbol{#1}}
\newcommand{\E}[1]{\mathbb{E}\left\{#1\right\}}

 \newcommand{\vg}{\ve{g}}
\newcommand{\vH}{\ve{H}} \newcommand{\vh}{\ve{h}}

\newcommand{\vW}{\ve{W}}

\newcommand{\qg}{{\bf g}}
\newcommand{\qh}{{\bf h}}

\newcommand{\qn}{{\bf n}}

\newcommand{\qw}{{\bf w}}

\newcommand{\qA}{{\bf A}}

\newcommand{\qH}{{\bf H}}
\newcommand{\qI}{{\bf I}}

\newcommand{\qW}{{\bf W}}

\newcommand{\SINRuA}{\mathsf{SINR_u^A}}
\newcommand{\SINRdA}{\mathsf{SINR_d^A}}
\newcommand{\SINRuS}{\mathsf{SINR_u^S}}
\newcommand{\SINRdS}{\mathsf{SINR_d^S}}
\newcommand{\SNRd}{\mathsf{SNR_d}}
\newcommand{\SNRu}{\mathsf{SNR_u}}

\newcommand{\HudS}{\qH_{\mathsf{ud}}^{pq}}
\newcommand{\HudSdag}{\qH_{\mathsf{ud}}^{pq\dag}}

\newcommand{\Us}{\mathsf{u}}
\newcommand{\Ds}{\mathsf{d}}
\newcommand{\AP}{\mathsf{a}}

\newcommand{\Sap}{\sigma_{\AP\AP}^2}

\newcommand{\PD}{p_\mathsf{D}}
\newcommand{\MRC}{\mathsf{MRC}}
\newcommand{\MRT}{\mathsf{MRT}}
\newcommand{\ZF}{\mathsf{ZF}}
\newcommand{\FD}{\mathsf{FD}}
\newcommand{\RFD}{\mathcal{R}^\FD_{\mathsf{sum}}}
\newcommand{\RHDs}{R^\mathsf{{HD}}_{\mathsf{{sum}}}}

\newcommand{\Galf}{\mathcal{G}(M,\phi,\delta,\PD\lambda)}

\newcommand{\be}{\begin{equation}} \newcommand{\ee}{\end{equation}}
\newcommand{\bea}{\begin{eqnarray}} \newcommand{\eea}{\end{eqnarray}}

\usepackage{multibib}
\usepackage[nodisplayskipstretch]{setspace}
\newcites{Prim}{Very important papers}

\definecolor{light-gray}{gray}{0.65}

\newcounter{mytempeqcounter}

\begin{document}

\title{\fontsize{0.84cm}{1cm}\selectfont  Uplink and Downlink Rate Analysis of a Full-Duplex C-RAN with Radio Remote Head Association}

\author{Mohammadali Mohammadi$^*$, Himal A. Suraweera$^\dagger$, and Chintha Tellambura$^\ddagger$\\
\small{$^*$Faculty of  Engineering, Shahrekord University, Iran (e-mail: m.a.mohammadi@eng.sku.ac.ir)}\normalsize\\
\small{$^\dagger$Department of Electrical and Electronic Engineering, University of Peradeniya, Sri Lanka (e-mail: himal@ee.pdn.ac.lk)}\normalsize\\
\small{$^\ddagger$Department of Electrical and Computer Engineering, University of Alberta, Canada (e-mail: chintha@ece.ualberta.ca)}\normalsize}

\maketitle
\begin{abstract}
We characterize the uplink (UL) and downlink (DL) rates of a full-duplex cloud radio access network (C-RAN) with all participate and single best remote radio head (RRH) association schemes. Specifically, multi-antenna equipped RRHs distributed according to a Poisson point process is assumed. The UL and DL sum rate of the single best RRH association scheme is maximized using receive and transmit beamformer designs at the UL and DL RRHs, respectively. In the case of the single best strategy, we study both optimum and sub-optimum schemes based on maximum ratio combining/maximal ratio transmission (MRC/MRT) and zero-forcing/MRT (ZF/MRT) processing. Numerical results show that significant performance improvements can be achieved by using the full-duplex mode as compared to the half-duplex mode. Moreover, the choice of the beamforming design and the RRH association scheme have a major influence on the achievable full-duplex gains.
\end{abstract}

\section{Introduction}
\label{sec:intro}

Cloud radio access network (C-RAN) is a new conceptual framework for  implementing future wireless networks~\cite{Peng:WCL:2014,VPoor:SPL:2013,Ratnarajah:TSP:2015}. In this centralized architecture, distributed access points known as remote radio heads (RRHs) forward user signals to/from a baseband unit (BBU) via a high speed optical fronthaul link. Consequently, C-RAN solutions can overcome path loss effect and deploy centralized signal processing to manage   interference effectively~\cite{Liu:TSP:2014}.

On parallel, full-duplex communication has emerged as  a complementary approach  for 5G wireless  since it has the potential to double the spectral efficiency of 5G  wireless. In essence, full-duplex radio performs simultaneous transmit/receive operations at the same frequency. Research on full-duplex has progressed  rapidly on  a variety of aspects such as theory, design and hardware implementation with the promise of making it  a viable practical solution soon~\cite{Duarte:PhD:dis, Riihonen:JSP:2011}. To this end, a major performance-limiting factor   is the loopback interference (LI) experienced at the input of a full-duplex transceiver~\cite{Himal:WCOM:FD:2014,Mohammad:TCOM:2015}. In order to mitigate LI, antenna domain techniques such as the use of electromagnetic shields, directional antennas and antenna
separation schemes can be employed~\cite{Everett:TWC:2014}.  When full-duplex and C-RAN are combined, due to the distributed RRHs, path loss serves a simple effective phenomenon for LI mitigation.

In the current literature, stochastic-geometry tools  have been widely adopted to study the C-RAN performance. Assuming a Poisson point process (PPP) distributed RRHs, the ergodic capacities of two (single-nearest and $N$-nearest) RRH association schemes were characterized in~\cite{Peng:WCL:2014}. By considering beamforming and base station selection, in~\cite{VPoor:SPL:2013} the performance of distributed antenna arrays was characterized.  In \cite{Ratnarajah:TSP:2015}, the DL performance of a multiple antenna equipped C-RAN with maximal ratio transmission (MRT) was analyzed. In~\cite{Liu:TSP:2014}, a  C-RAN  was optimized via  DL antenna selection and regularized zero forcing (ZF).  Deviating from the existing body of work that has only focused on uplink (UL) or DL performance, \cite{Leung:TVT} considered a full-duplex distributed antenna relay implementation. However, it assumes perfect LI cancellation. Thus, many theoretical questions remain open.

In this paper, we consider a C-RAN architecture in which multiple antenna equipped RRHs communicate with a full-duplex user to support simultaneous UL and DL transmission. Our contributions are summarized as follows:

\begin{itemize}
\item Assuming optimum,  maximum ratio combining (MRC)/MRT, and ZF/MRT beamforming designs, we derive exact and tractable expressions for the average UL and DL rate of the full-duplex user for the single best UL/DL RRH association (SRA) scheme.
\item Our results reveal that for a fixed value of LI power, the optimum and ZF/MRT schemes can ensure a balance between maximizing the system average sum rate and maintaining acceptable level of fairness between the UL/DL transmission. Moreover, the performance of MRC/MRT can be substantially improved by adopting an appropriate DL and UL association scheme.
\end{itemize}
\emph{Notation:} We use lower case/upper case bold letters to denote vectors and matrices, respectively. $\|\cdot\|$, $(\cdot)^{\dag}$, $(\cdot)^{-1}$ and ${\mathsf{trace}}(\cdot)$ denote the Euclidean norm, conjugate transpose operator, matrix inverse and trace of a matrix respectively; ${\tt E}\left\{x\right\}$ stands for the expectation of random variable (RV) $x$; $F_X(\cdot)$ denote the associated  cumulative distribution function (cdf)  and $\mathcal{M}_X(s)$, the moment generating function (MGF).  $\Gamma(a)$ is the Gamma function; $\Gamma(a,x)$ is upper incomplete Gamma function~\cite[Eq. (8.310.2)]{Integral:Series:Ryzhik:1992}; and $G_{p, q}^{m, n} \left( z \  \vert \  {a_1\cdots a_p \atop b_1\cdots b_q} \right)$ denotes the Meijer G-function~\cite[ Eq. (9.301)]{Integral:Series:Ryzhik:1992}.

\section{System Model}\label{sec:SYS Model}
Consider a C-RAN consisting of a  BBU and  a group of spatially distributed RRHs each having $M \geq 1$  antennas jointly support UL and DL transmissions for a full-duplex user, denoted by $U$. We assume that the full-duplex user is equipped with two antennas: one receive antenna and one transmit antenna. The locations of the RRHs are modeled as a homogeneous PPP $\Phi=\{x_k\}$ with density $\lambda$ in a disc $\mathcal{D}$, of radius $R$. We assume that $100\PD\%$ of the RRHs, are deployed to assist the DL communication and $100(1-\PD)\%$ for UL communication. Therefore, the set of DL RRHs is denoted as $\Phi_{\Ds} =\{x_k\in \Phi: B_k(\PD)=1\}$ where $B_k(\PD)$  are independent and identically distributed (i.i.d.) Bernoulli RVs with parameter $p$ associated with $x_k$. Similarly, the set of UL RRHs is a PPP with density $(1-\PD)\lambda$ and is denoted as $\Phi_{\Us} =\{x_k\in \Phi: B_k(\PD)=0\}$.

\subsection{Channel Model}
The channel model consist of both small-scale multipath fading and large-scale path loss. We denote the DL channel vector from RRH $i$ to $U$ as $\qh_i\in \mathbb{C}^{M \times 1}$ and the UL channel vector from $U$ to RRH $i$ as $\qg _i^{\dag}\in \mathbb{C}^{1 \times M}$, respectively.  These channels capture the small-scale fading and are modeled as Rayleigh fading such that $\qg _i$  and $\qh_i \sim  \mathcal{CN}(\textbf{0}_M, \qI_M )$,  where $\mathcal{CN}(\cdot,\cdot )$,  denotes a circularly symmetric complex Gaussian distribution. The path loss function is given by $\ell(x_1,x_2)=\|x_1-x_2\|^{-\mu}$, with $\mu>2$ is the path loss exponent. Further, as in \cite{Ratnarajah:TSP:2015} we assume that there exist an ideal low-latency backhaul network with sufficiently large capacity (e.g. optical fiber) connecting the set of RRHs to the BBU, which performs all the baseband signal processing and transmission scheduling for all RRHs.

\subsection{Association Schemes}
For this  system, we investigate the performance of the following two RRH association schemes:
\begin{itemize}
\item \emph{All RRH Association (ARA) Scheme}: All corresponding DL RRHs cooperatively transmit the signal, $s_{\Ds}$ to the full-duplex user, $U$. Moreover, all the corresponding UL RRHs deliver signals from $U$ to the BBU.
\item \emph{SRA Scheme}:  UL RRH and DL RRH with the best channels from/to $U$ is selected in order to receive and transmit UL/DL signals. We also model that an interference region (IR) is adopted by the BBU to protect the UL RRH against interference from the DL RRH. No DL RRH transmission is allowed within the IR and $U$ associates with the DL RRH having the best channel strength within the selection region $\mathcal{A}$.  Without loss of generality, we assume that $U$ is located at the origin of $\mathcal{D}$~\cite{Peng:WCL:2014,Ratnarajah:TSP:2015}. Moreover, let us denote $\qw_{t,i}\in\mathbb{C}^{M\times 1}$ as the transmit beamforming vector at the DL RRH $i$ and $\qw_{r,i}\in\mathbb{C}^{M\times 1}$ as the receive beamforming vector at the UL RRH, $i$. Therefore, the associated UL RRH $p$ and DL RRH $q$ for user $U$ are given by
\begin{subequations}
\begin{align}
p&=\argmax_{i\in \Phi_{\mathsf{u}}}\{\ell(x_{i})\|\qw_{r,i}^\dag\qg_i\|^2\} \\
q&=\argmax_{i\in \Phi_{\mathsf{d}}\bigcap \mathcal{A}}\{\ell(x_{i})\|\qh_i^{\dag}\qw_{t,i}\|^2\}.
\end{align}
\end{subequations}
\end{itemize}

In this paper we consider a sectorized IR of angle $\pm\phi$ around the $U-p$ axis. As shown in Section~\ref{sec:numerical} the UL/DL sum rate performance will be dependent on $\phi$.
\subsection{Uplink/Downlink Transmission}
\emph{DL Transmission:}
Similar to \cite{Ratnarajah:TSP:2015}, we assume that all DL RRHs transmit with power $P_b$. Hence, according to the ARA scheme, the received signal at $U$ can be expressed as
\be\label{eq:rx at user}
y_\Ds  = \sum_{i \in \Phi_\Ds} \sqrt{P_b \ell(x_{i})}\qh_{i}^{\dag}\qw_{t,i} s_{\Ds} \!+\! \sqrt{P_u}h_{\mathsf{LI}}s_u + n_\Ds,
\ee
where $P_u$ is the user transmit power, $s_u$  is the user signal satisfying ${\tt E}\left\{s_u s_u^{\dag}\right\}=1$, and $n_\Ds$ denotes the additive white Gaussian noise (AWGN). We proceed with all noise variances set to one.
$h_{\mathsf{LI}}$ denotes the LI channel at the user.  In order to mitigate the adverse effects of the LI on system's performance, an interference cancellation scheme (i.e. analog/digital cancellation) can be used at the full-duplex user and we model the residual LI channel with Rayleigh fading assumption since the strong line-of-sight component can be estimated and removed~\cite{Duarte:PhD:dis}. Since each implementation of a particular analog/digital LI cancellation scheme can be characterized by a specific residual power, a parametrization by $h_{\mathsf{LI}}$ satisfying $\E{|h_{\mathsf{LI}}|^2}=\Sap$  allows these effects to be studied in a generic way~\cite{Riihonen:JSP:2011}.

By invoking~\eqref{eq:rx at user}, the DL signal-to-interference-plus-noise ratio (SINR) at the $U$ with ARA and SRA schemes are respectively expressed as
\begin{subequations}
\begin{align}
 \SINRdA &= \frac{\sum_{i \in \Phi_\Ds} P_b \ell(x_{i})\|\qh_{i}^{\dag}\qw_{t,i}\|^2  }{{P_u}|h_{\mathsf{LI}}|^2 + 1}\label{eq:SINRd AR},\\
  \SINRdS&= \frac{P_b \ell(x_{q})\|\qh_{q}^{\dag}\qw_{t,q}\|^2  }{{P_u}|h_{\mathsf{LI}}|^2+ 1}.\label{eq:SINRd SR}
\end{align}
\end{subequations}
\emph{UL Transmission:} In the considered full-duplex C-RAN, UL transmission is impaired by the inter-RRH interference due to DL RRH transmission. Therefore, in case of the ARA scheme, received signal at the BBU is given by
\begin{align}\label{eq:rx at user}
 y_u&=\sum_{j\in \Phi_{\mathsf{u}} } \Big(\sqrt{P_u\ell(x_{j})}\qw_{r,j}^{\dag}\qg_{j} s_u
 \\&\quad
 +\sum_{i\in \Phi_{\mathsf{d}} \bigcap \mathcal{A} } \sqrt{P_b \ell(x_j,x_{i})}\qw_{r,j}^{\dag}\qH_{\mathsf{ud}}^{ji}\qw_{t,i}s_{\Ds}+ \qw_{r,j}^{\dag}\qn_j\Big),\nonumber
\end{align}
where $\qH_{\mathsf{ud}}^{ji}\in \mathbb{C}^{M \times M}$  is the channel matrix between the DL RRH $i$ and UL RRH $j$  consists of complex Gaussian distributed entries with zero mean and unit variance, $\qn_j\sim  \mathcal{CN}(\textbf{0}_M, \qI_M )$  denotes the AWGN vector at the  UL RRH $j$. According to the SRA scheme only one UL (best) RRH and one DL (best) RRH are selected to assist the full-duplex user. Let the sub-indexes $p$ and $q$ correspond to the active UL and DL RRH, respectively. Therefore, the SINR with ARA and SRA schemes can be respectively expressed as
\vspace{-0.2em}
\begin{subequations}
\begin{align}
 \SINRuA&=
 \frac{\sum_{j\in \Phi_{\mathsf{u}}} P_u \ell(x_{j})\|\qw_{r,j}^{\dag}\qg_{j}\|^2}
                       { I_{\mathsf{ud}}+ \|\qw_{r,j}\|^2},\label{eq:SINRu AR}\\
 \SINRuS &=  \frac{P_u \ell(x_{p})\|\qw_{r,p}^{\dag}\qg_{p}\|^2  }
{ P_b \ell(x_p,x_{q})\|\qw_{r,p}^{\dag}\qH_{\mathsf{ud}}^{pq}\qw_{t,q}\|^2+ \|\qw_{r,p}\|^2},\label{eq:SINRu SR}
\end{align}
\end{subequations}
where $I_{\mathsf{ud}}=\sum_{j\in \Phi_{\mathsf{u}}}\:\sum_{i\in \Phi_{\mathsf{d} \bigcap \mathcal{A}}}\! P_b \ell(x_j,x_{i})|\qw_{r,j}^{\dag}\qH_{\mathsf{ud}}^{ji}\qw_{t,i}|^2.$

In the next section, we consider different transmit and receive beamforming vector designs and characterize the system performance in terms of the average UL and DL sum rate given by
\vspace{-0.5em}
\begin{align}
\RFD  = \mathcal{R}_\Us + \mathcal{R}_\Ds,\label{eq:achievable rate FD}
\end{align}
where \small{$\mathcal{R}_\Us =\E{\ln\left(1+\mathsf{SINR_u^i}\right)}$, $\mathcal{R}_\Ds =\E{ \ln \left(1+\mathsf{SINR_d^i}\right)}$ }\normalsize with $\mathsf{i}\in\{\mathsf{A}, \mathsf{S}\}$ are the spatial average UL and DL rates, respectively.
\section{Joint Precoding/Decoding Designs}
We now consider several transmit/receive beamformer designs to suppress/cancel the inter-RRH interferences.  Specifically, we present the optimal design that maximizes the achievable sum rate of the SRA scheme. We further investigate the ZF/MRT and MRC/MRT suboptimal beamforming designs, where the former is applicable for the SRA scheme. Each of these designs offers a different performance-complexity tradeoff.
\subsection{The Optimal Processing}
In this subsection, our main objective is to jointly design the transmit and receive beamformers at the selected DL and UL RRH pair so that system achievable sum rate in~\eqref{eq:achievable rate FD} is maximized. Specifically, the sum rate  maximization problem can be formulated as
 \begin{align}\label{eqn:opt problem}
    \max_{\qw_{t,q}, \qw_{r,p}} \hspace{1em}&
\RFD=\ln\left(1 +
    a_1 \|\qh_{q}^{\dag}\qw_{t,q}\|^2
    \right) \nonumber\\
&\hspace{1em}+
\ln\left(1 +
\frac{ a_2\|\qw_{r,p}^{\dag}\qg_{p}\|^2  }
 { a_3\|\qw_{r,p}^{\dag}\qH_{\mathsf{ud}}^{pq}\qw_{t,q}\|^2+ \|\qw_{r,p}\|^2}
 \right),\nonumber\\
    \mbox{s.t.} \hspace{1em}& \|\qw_{r,p}\|=\|\qw_{t,q}\|=1,
 \end{align}
where $a_1 = \frac{P_b \ell(x_{q})}{{P_u}|h_{\mathsf{LI}}|^2+ 1} $, $a_2 =P_u \ell(x_{p})$, and $a_3 = P_b \ell(x_p,x_{q})$. In order to solve the problem in~\eqref{eqn:opt problem}, we first fix $\qw_{t,q}$ and optimize $\qw_{r,p}$ to maximize $\RFD$. Note that given $\qw_{t,q}$, $\qw_{r,p}$ only influence the achievable UL rate.  Therefore, using the fact that logarithm is a monotonically  increasing function, the optimization problem can be written as
\begin{align}\label{eq:opt: UL problem}
\underset{\|\qw_{r,p}\|^2 =1}
{\text{max}}\hspace{1em}&\frac{ a_2\|\qw_{r,p}^{\dag}\qg_{p}\|^2  }
 { a_3\|\qw_{r,p}^{\dag}\qH_{\mathsf{ud}}^{pq}\qw_{t,q}\|^2+ \|\qw_{r,p}\|^2},
\end{align}
which is a generalized Rayleigh ratio problem. It is well known that~\eqref{eq:opt: UL problem} is globally maximized
when
\begin{align}\label{eq:opt:w_r}
\qw_{r,p} = \frac{\left(a_3\qH_{\mathsf{ud}}^{pq}\qw_{t,q}\qw_{t,q}^{\dag}\qH_{\mathsf{ud}}^{pq\dag} + \qI\right)^{-1}\qg_{p}}
{\Big\| \left(a_3\HudS\qw_{t,q}\qw_{t,q}^{\dag}\HudSdag + \qI\right)^{-1}\qg_{p} \Big\|}.
\end{align}
Accordingly, by substituting $\qw_{r,p}$ into~\eqref{eq:opt: UL problem} and applying the Sherman Morrison formula, the optimization problem in~\eqref{eqn:opt problem} can be re-formulated as
 \begin{align}\label{eqn:opt problem reform}
    &\max_{\|\qw_{t,q}\|^2=1} \hspace{1em}
\RFD=\ln\left(1 +
    a_1 \|\qh_{q}^{\dag}\qw_{t,q}\|^2
    \right) \\
&\hspace{0em}+
\ln\left(1 +
{a_2}\left(\|\vg_{p}\|^2 - \frac{a_3 \|\vg_{p}^{\dag}\HudS \qw_{t,q}\|^2 }
{1 + a_3 \qw_{t,q}^{\dag}\HudSdag\HudS \qw_{t,q} }\right)
 \right),\nonumber
 \end{align}
 which is still difficult to solve due to its nonconvex nature. To solve the problem in~\eqref{eqn:opt problem reform}, we apply a similar approach as in~\cite{Mohammad:TCOM:2015} to convert the optimization problem to
\begin{align}\label{eqn:opt problem SDP}
\underset{ \qw_{t,q}}{\text{max}}&
\quad{\mathsf{trace}}(\vh_{q}^{\dag}{\vW}_{t}\vh_{q})\nonumber\\
\text{s.t.}&
\quad{\mathsf{trace}} ({\vW}_{t}(\HudSdag \vg_{p}\vg_{p}^{\dag}\HudS - \mu\:\HudSdag\HudS)) = \frac{\alpha}{a_3},\nonumber\\
&\quad
{\vW}_{t} \succeq 0, \quad{\mathsf{trace}}({\vW}_{t}) =1, \quad\text{rank}({\vW}_{t})=1,
\end{align}
where $\alpha\!=\!\frac{ a_3\|\vg_{p}^{\dag}\HudS \qw_{t,q}\|^2} {1 + a_3\qw_{t,q}^{\dag}\HudSdag\HudS \qw_{t,q}}$ and ${\vW}_{t} = \qw_{t,q}\qw_{t,q}^{\dag}$ is a  symmetric, positive semi-definite matrix. In order to solve~\eqref{eqn:opt problem SDP}, we can resort to the widely used semidefinite relaxation approach. By dropping the rank-1 constraint, the resulting problem becomes a semidefinite program, whose solution ${\qW}_{t}$ can be found by using the method provided in~\cite[Appendix B]{Zheng:JSPL2013} or by using appropriate solvers, for example, CVX.

Denoting the optimal objective value of~\eqref{eqn:opt problem SDP} as $f(\alpha)$, the achievable sum rate maximization
problem can be formulated as
\begin{align}\label{eqn:opt problem 1-D}
\underset{\alpha \geq 0}{\text{max}}&
\quad \RFD (\alpha) =
\log_2\left(\left(1\! +\! a_1 f(\alpha)\right) \left(1 + a_2\left(\|\vg_{p}\|^2 \!- \!\alpha\right)\right)\right).
\end{align}
Therefore, in order to solve~\eqref{eqn:opt problem}, it remains to perform a one-dimensional optimization with respect to $\alpha$.

\subsection{ZF/MRT Processing}
As a suboptimal design, we can adopt ZF at the UL RRH to completely cancel inter-RRH interference with SRA scheme~\cite{Himal:WCOM:FD:2014}. To ensure feasibility, the number of antennas equipped at the UL RRH should be greater than one, i.e., $M>1$. After substituting $\qw_{t,q}^{\MRT} = \frac{\qh_q}{\|\qh_q\|}$  into~\eqref{eq:SINRu AR}, the optimal receive beamforming vector at the UL RRH $\qw_{r,p}$ can be obtained by solving the following problem:
   \bea\label{eqn:wt}
    \max_{\|\qw_{r,p}\|=1} &&\hspace{1em}  \|\qw_{r,p} \qg_{p}\|^2 \nonumber\\
     \mbox{s.t.} &&\hspace{1em} \qw_{r,p}^{\dag}\qH_{\mathsf{ud}}^{pq}\qh_q =0.
 \eea
Hence, the optimal combining vector $\qw_{r,p}$  can be obtained as
\bea\label{eqn:wtZF}
\qw_{r,p}^{\ZF} = \frac{\qA \qg_{p}}{\|\qA \qg_{p}\|},
 \eea
where $\qA\triangleq \qI - \frac{\qH_{\mathsf{ud}}^{pq}\qh_q \qh_q^{\dag}\qH_{\mathsf{ud}}^{pq \dag} }{ \| \qH_{\mathsf{ud}}^{pq}\qh_q\|^2}$.

\subsection{MRC/MRT Processing}
In addition, we also consider a MRC/MRT suboptimal beamforming design. Although MRC/MRT processing is not optimal in presence of inter-RRH interference, it could be favored in practice, because it can balance the performance and system complexity.
For the MRC/MRT scheme, $\qw_{r,p}$ and $\qw_{t,q}$ are set to match the UL and DL channels, respectively. Hence, $\qw_{r,p}^{\MRC}=\frac{\qg_p}{\|\qg_p\|}$ and $\qw_{t,q}^{\MRT}=\frac{\qh_q}{\|\qh_q\|}$.

\section{Performance Analysis}
In this section, average UL/DL rates of the considered RRH association schemes together with ZF/MRT and MRC/MRT processing are evaluated. In case of the optimal scheme, derivation of the UL/DL rates are difficult and we use simulations in Section~\ref{sec:numerical}. Moreover, deriving the statistics of the SINR in \eqref{eq:SINRu AR} for the ARA scheme with MRC/MRT and ZF/MRT processing seems to be an intractable task. Hence, in order to evaluate the average UL rate, we have resorted to simulations in Section~\ref{sec:numerical}.

\subsection{ZF/MRT Scheme}
By substituting $\qw_{r,p}^{\ZF}$ and $\qw_{t,q}^{\MRT}$ into~\eqref{eq:SINRd SR} and~\eqref{eq:SINRu SR} the received SINR at the user and the BBU are obtained as $\SINRdS= \frac{P_b \ell(x_{q})\|\qh_{q}\|^2}{{P_u}|h_{\mathsf{LI}}|^2+ 1}$ and $ \SINRuS =  P_u \ell(x_{p})\|\tilde{\qg}_{p}\|^2$, respectively, where $\tilde{\qg}_{p}$ is a $(M-1)\times 1$ vector~\cite{Mohammad:TCOM:2015}. Note that according to the choice of the IR parameter $\phi$, there is a probability $p_\emptyset$ that the DL RRH set (and thus interfering set) is empty. In this case ZF beamformer at the UL RRH reduces to MRC beamformer and thus $ \SINRuS =  P_u \ell(x_{p})\|{\qg}_{p}\|^2$.

As a preliminary, we first present the cdf of the RV $X = \max_{k\in\Phi_{\Ds} \bigcap \mathcal{A}}\{\ell(x_{k})\|\qh_{k}\|^2\}$, which will be invoked in the subsequent derivations. For notation convenience, we define $\delta=\frac{2}{\mu}$.
\begin{Lemma}\label{lemma:cdf X}
Let $\delta= \frac{m}{n}$ with $\gcd(m,n)=1$ where $\gcd(m,n)$ is the greatest common divisor of integers $m$ and $n$. The exact cdf of $X$ is given by
\begin{align}\label{eq:cdf of X}
F_{X}(t) =\exp\left(-\Galf t^{-\delta}\right),
\end{align}
where $\Galf = \frac{\PD\lambda(\pi-\phi)}{\Gamma(M)}\Gamma\left(M+\delta\right)$.
\end{Lemma}
\begin{proof}
The proof is omitted due to space limitations.
\end{proof}

We now present average sum rate with ZF/MRT processing in the following proposition.

\begin{proposition}\label{Propos:Rc:ZF}
The average sum rate achieved by the SRA scheme with ZF/MRT processing is expressed by~\eqref{eq:achievable rate FD}
where
 \begin{align}\label{eq:Ru ZF/MRT}
 \mathcal{R}_\Us &= e^{-\lambda\PD (\pi\!-\phi) R^2} \mathcal{R}_{M} + (1\!-\!e^{-\lambda\PD (\pi-\phi) R^2}) \mathcal{R}_{(M-1)},
\end{align}
where
 \begin{align}\label{eq:Ru ZF}
 \mathcal{R}_{M}\!=\!\eta_\Us^\ZF
 G_{v, 2m}^{2m,  t} \left( \varsigma_\Us^\ZF  \Big\vert \  {\Delta(n,0),\Delta(m,0),\Delta(m,1) \atop \Delta(m,0), \Delta(m,0) }\right)\!\!,
\end{align}
with $\eta_\Us^\ZF\!\!=\!\!\sqrt{\frac{n}{(2\pi)^{v-3}}}$, $\varsigma_\Us^\ZF\!\!=\!\!\left(\frac{n}{\mathcal{G}(M,0,\delta,(1-\PD)\lambda)P_u^{\delta}}\right)^n $, $t\!=\!m+n$, $v\!=\!2m+n$, and $\Delta(a,b) =\frac{b}{a},\cdots,\frac{a+b-1}{a}$.

Moreover, the average DL rate is given by
\begin{align}\label{eq:Rd ZF}
 \mathcal{R}_\Ds
 &\!=\left(\int_{0}^{\infty} \left(1 \!-\! \eta_\Ds^\ZF G_{t, 0}^{0, t}
\left( \varsigma_\Ds^\ZF\left(\frac{m}{z}\right)^m \!\  \Big\vert \ \! {\Delta(m,0),\Delta(n,1) \atop -}\!\right)\right)\right.\nonumber\\
&\left.\hspace{2em}\times\frac{\exp(-z)}{z(1 +P_u\Sap z)}dz\right)\left(1-e^{-\lambda \PD (\pi\!-\phi) R^2}\right),
\end{align}
where $\eta_\Ds^\ZF =\sqrt{\frac{mn}{(2\pi)^{t-2}}}$ and $\varsigma_\Ds^\ZF = \left(\frac{n}{\Galf P_b^{\delta}}\right)^n$.
\end{proposition}

\begin{proof}
The proof is omitted due to space limitations.
\end{proof}

\subsection{MRC/MRT Scheme}
For DL transmission, the receive SINR of the MRC/MRT and ZF/MRT processing are the same. Therefore, the average DL rate of MRC/MRT is given by~\eqref{eq:Rd ZF}. Now, we derive the the average UL rate of MRC/MRT processing.

Substituting $\qw_{r,p}^{\MRC}$ and $\qw_{t,q}^{\MRT}$ into~\eqref{eq:SINRu SR}, the received SINR at the BBU can be expressed as
\begin{align}\label{eq:SINRu SR2}
 \mathsf{SINR_u^S} &=  \frac{P_u \ell(x_{p})\|\qg_{p}\|^2}
 { P_b\ell(x_p,x_{q})\sum_{i=1}^M Z_i+ 1},
\end{align}
where $Z_i  = U_i V_i$  with $U_i  = |\qw_{r,p}^{\MRC^\dag}\qh_{\mathsf{ud}i}^{pq}|^2$ and $V_i = (w_{t,q,i}^{\MRT})^2$ where $\qh_{\mathsf{ud}i}^{pq}$ is the $i$th column of $\vH_{\Us\Ds}^{pq}$ (i.e., $\vH_{\Us\Ds}^{pq} = [\vh_{\mathsf{ud}1}^{pq},\vh_{\mathsf{ud}2}^{pq},\cdots,\vh_{\mathsf{ud}M}^{pq}]$) and $w_{t,q,i}^{\MRT}$ is the $i$th element of $\qw_{t,q}^{\MRT}$.
For the notational convenience, let us denote $W = P_u \ell(x_{p})\|\qg_{p}\|^2 $, and $Z =P_b\ell(x_p,x_{q})\sum_{i=1}^M Z_i $.

Note that the MGF of $W$ follows from~\eqref{eq:cdf of X} by making the substitution of corresponding parameters, i.e., $\phi\rightarrow0$ and $\PD\rightarrow(1-\PD)$ and then using the differentiation property of the Laplace transform. We now characterize the cdf of $Z_i$ in the following lemma which will be used to establish the average UL rate due to MRC/MRT processing.

\begin{Lemma}~\label{lemma:cdf Zi}
The exact cdf of $Z_i$ can be expressed as
\begin{align}\label{eq:cdf Zi}
F_{Z_i}(t) & =    G_{3 4}^{3 1} \left( \Sap t \  \Big\vert \  {1, M, M \atop 1, 1, M, 0} \right).
\end{align}
\end{Lemma}

\begin{proof}
The proof is omitted due to space limitations.
\end{proof}

\begin{proposition}\label{Propos:Rc:MRC}
The average UL rate achieved by the SRA scheme with MRC/MRT processing is expressed as
\begin{align}\label{eq:Ru:MRC}
&\hspace{0em}\mathcal{R}_\Us=
\int_{0}^{\infty}
\int_0^{2R}
\left(G_{4 4}^{3 2} \left( \frac{\Sap r^{\alpha}}{P_bz} \  \Big\vert \  {0,1, M, M \atop 1, 1, M, 0} \right)\right)^{M}
\nonumber\\
&\hspace{2em}\times
\left(1-\!\eta_\Us^\MRC G_{t, 0}^{0, t}
\left(\varsigma_\Us^\MRC\!\left(\frac{m}{s}\right)^m \! \ \Big\vert \ \!\! {\Delta(m,0),\Delta(n,1) \atop -}\!\right)\!\right)\nonumber\\
&\hspace{2em}\times\frac{\exp(-z)}{z} f_{r}(r) dr dz\!,
\end{align}
where $\eta_\Us^\MRC = \eta_\Ds^\ZF$, $\varsigma_\Us^\MRC=\left(\frac{n}{\mathcal{G}(M,0,\delta,(1-\PD)\lambda)P_u^{\delta}}\right)^n $ and $f_{r}(r)$ is given in~\cite{Moltchanov:Distance}.

\end{proposition}
\begin{proof}
The proof is omitted due to space limitations.
\end{proof}

\begin{proposition}
The average DL rate achieved by the ARA scheme with MRC/MRT and ZF/MRT processing can be expressed as
\begin{align}\label{eq:Rd singular}
\mathcal{{R}}_\Ds &\!=\!\sum_{k=1}^{\infty}\frac{( \delta\Gamma(-\delta)P_b^{\delta}\Galf )^{k}}{\Gamma(k+1)} \nonumber\\
& \qquad\times G_{2 1}^{1 2} \left( P_u \Sap  \!\!\  \Big\vert \ \!\! {1 \!-\! \delta k,0 \atop 0} \right).
\end{align}
\end{proposition}

\begin{proof}
The proof is omitted due to space limitations.
\end{proof}

\section{Numerical Results and Discussion}\label{sec:numerical}
We now present several numerical examples of full-duplex C-RAN performance. The maximum transmit power of the DL RRHs and the full-duplex user are set to $23$ dBm. The power spectral density of noise is set as $-120$~dBm/Hz. We assume that $R=150$ m, $\alpha=3$ and $\lambda=0.001$.
\begin{figure}[t]
\centering
\vspace{0em}
\includegraphics[width=88mm, height=67mm]{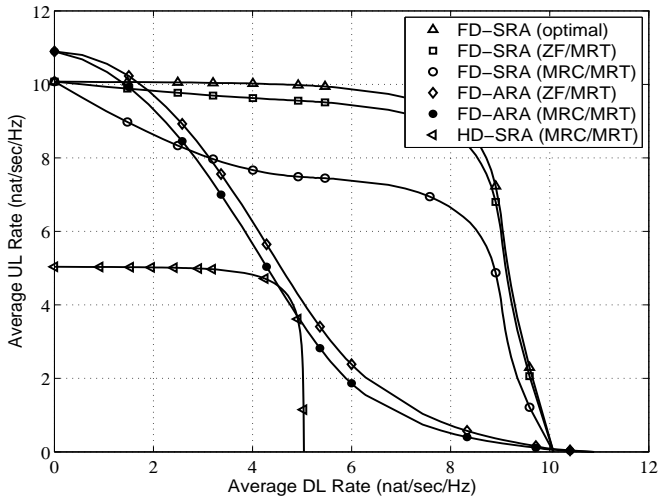}
\caption{Rate region of the ARA and SRA schemes for full-duplex and half-duplex modes of operation ($M=3$ and $\phi=\pi/3$).}
\label{fig: rate region}
\end{figure}

Fig.~\ref{fig: rate region} shows the
rate region of the ARA and SRA schemes for both full-duplex
and half-duplex modes of operation. A half-duplex user employs orthogonal UL and DL time slots for operation. Consequently, with the ARA scheme and MRC/MRT precessing, the average sum rate of the half-duplex user is given by
\begin{align*}\label{eq: sum rate of single-antenna HD AP}
\RHDs&\!=\tau {\tt E}\{\ln(1 \!+\! \SNRd )\} + (1\!-\!\tau){\tt E}\{\ln(1 \!+\!\SNRu  )\},
\end{align*}
where $\tau$ is a fraction of the time slot duration of $T$, used for DL transmission, \small{$\SNRd=\sum_{i \in \Phi_\Ds } P_b \ell(x_{i})\|\qh_{i}^{\dag}\qw_{t,i}\|^2  $ }\normalsize and \small{$\SNRu=\sum_{j\in \Phi_{\mathsf{u}}} P_u \ell(x_{j})\|\qw_{r,j}^{\dag}\qg_{j}\|^2$. }\normalsize We have set $P_u=23$ dBm, $P_b=23$ dBm, $\Sap=-40$ dBm, and $\tau=0.5$ in Fig.~\ref{fig: rate region} and change $\PD$ from zero (only UL transmission) to one (only DL transmission).

For the ARA scheme with ZF/MRT processing we assume that each UL RRH adjusts its receive beamforming vector in such a way that the interference from its nearest DL RRH is canceled. These results reveal that the rate region of the ARA scheme is strongly biased towards UL or DL. However, the SRA scheme can guarantee a more balanced rate region. For this setup, SRA scheme with the optimal beamforming design and ZF/MRT can achieve up to $89\%$ and $80\%$ average sum rate gains as compared to the half-duplex SRA counterparts, respectively. Our observation of the relation between the rate region of MRC/MRT and IR parameter $\phi$ (which is not shown for the sake of clarity) shows that there is an optimal $\phi$ that tends MRC/MRT rate region towards the ZF/MRT one.

Fig.~\ref{fig: sum rate rate} shows the impact of the IR region parameter $\phi$ on the sum rate of different beamformer designs at the DL and UL RRHs and for the SRA scheme. Intuitively, increasing the $\phi$ (shrinking the selection region) decreases the number of DL RRHs and consequently the DL rate. Moreover, the UL rates of optimum and ZF/MRT designs remain constant to produce an overall sum rate decrease as $\phi$ is increased.
On the contrary, increasing $\phi$ improves the performance of MRC/MRT  because  the inter-RRH interference between the selected UL RRH and DL RRH is reduced.
Clearly, increasing $\phi$ beyond its optimum value does not improve the sum rate of MRC/MRT processing due to the fact that there may not be sufficient number of DL RRH inside the selection region.
\section{Conclusion}
We studied the average sum rate of a C-RAN in which spatially distributed multi-antenna RRHs are used to receive and transmit signals to a full-duplex user. Our analysis considered optimum beamformer design at the UL and DL RRHs as well as suboptimum MRC/MRT and ZF/MRC processing for the SRA scheme. Analytical expressions for the average DL rate of the suboptimum schemes with SRA and ARA schemes were derived, while the UL rate for the SRA scheme was obtained. For a fixed value of LI power, the SRA scheme with optimal and ZF/MRT processing can ensure a balance between maximizing the average sum rate and maintaining an acceptable UL/DL transmission fairness level. The performance of MRC/MRT processing can be substantially improved by optimally tuning the parameter $\phi$.
\begin{figure}[t]
\centering
\includegraphics[width=88mm, height=67mm]{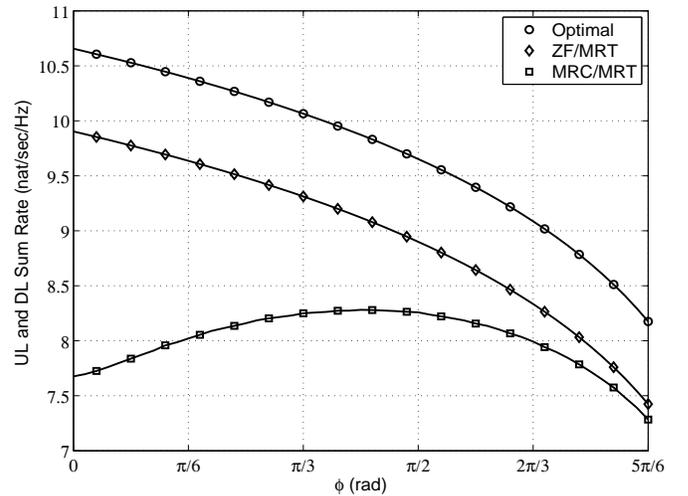}
\caption{ Average sum rate versus $\phi$ with different beamforming designs ($M=2$,  $P_u=10$ dBm, $P_b=10$ dBm, and $\Sap=-30$ dBm).}
\label{fig: sum rate rate}
\end{figure}


\bibliographystyle{IEEEtran}

\end{document}